\documentclass{amsart}

\usepackage[table]{xcolor}
\usepackage{hyperref}
\usepackage{enumerate}
\usepackage{upgreek}
\usepackage{bm}
\usepackage{pstool}
\usepackage{url}
\usepackage{amssymb}
\usepackage{amsthm}
\usepackage{tikz}

\theoremstyle{definition}

\newtheorem{theorem}{Theorem}[section]
\newtheorem{definition}[theorem]{Definition}

\newtheorem{prop}[theorem]{Proposition}

\newtheorem{corollary}[theorem]{Corollary}

\newcommand{\realnumb}{\mathbb{R}}
\newcommand{\CA}[1]{\mathfrak{Cs}\, #1\,}
\newcommand{\ca}[1]{\mathsf{Conc}\,#1\,}
\newcommand{\defeq}{\ \stackrel{\textsf{\tiny def}}{=} \ }
\newcommand{\defiff}{\stackrel{\textsf{\tiny def}}{\iff}}
\newcommand{\ie}{i.e., }
\newcommand{\eg}{e.g., }

\newcommand{\la}{\langle}
\newcommand{\ra}{\rangle}
\newcommand{\Poi}{\mathsf{Poi}}
\newcommand{\Triv}{\mathsf{Triv}}
\newcommand{\Scal}{\mathsf{Scal}}
\newcommand{\defequiv}{\stackrel{\scriptscriptstyle{\Delta}}{\equiv}}
\newcommand{\euclg}{\mathcal{E}\!\mathit{ucl}}
\newcommand{\cng}[1]{\bm{{ \cong}}_{#1}}
\newcommand{\cngrel}[1]{\mathrel{\bm{{ \cong}}_{#1}}}

\newcommand{\concept}{\mathsf{C}}
\newcommand{\llcon}{\bm{\uplambda}}
\newcommand{\abstime}{{\mathsf{S}}}
\newcommand{\rest}{\mathsf{Rest}}
\newcommand{\restrel}{\mathrel{\rest}}
\newcommand{\abstimerel}{\mathrel{\abstime}}

\newcommand{\orrigo}{\vec{\textsf{o}}}
\newcommand{\egyseg}{\vec{\textsf{e}}}
\newcommand{\Nodel}{{\mathfrak{N}}}
\newcommand{\Model}{{\mathfrak{M}}}

\newcommand{\group}{\mathsf{G}}

\newcommand{\aut}[1]{\mathsf{Aut}\; #1}

\newcommand{\fv}[1]{{#1}}
\newcommand{\relst}{\mathcal{R}\mathit{el}}
\newcommand{\RelST}{\mathcal{R}\mathit{elST}}
\newcommand{\GalST}{\mathcal{G}\mathit{alST}}
\newcommand{\LclST}{\mathcal{LC}\mathit{lassST}}
\newcommand{\rel}{\mathsf{R}}

\newcommand{\geometry}{\mathcal{G}}
\newcommand{\Col}{{\mathsf{Col}}}
\newcommand{\Bw}{{\mathsf{Bw}}}

\newcommand{\sqEl}[1]{\|#1\,\|^{2}}
\newcommand{\sqMl}[1]{\|#1\,\|_{\mu}^{2}}

\renewcommand{\qed}{\hfill$\Box$}

\title[On {A}ndr\'eka's Conjecture]{On {A}ndr\'eka's Conjecture that special relativity is the only possible conceptual reduct of classical kinematics} 
\author{Judit Madar\'asz}
\address{Judit Madar\'asz, HUN-REN Alfr\'ed R\'enyi Institute of Mathematics, Budapest, Hungary}
\email{madarasz.judit@renyi.hu}

\author{Mike Stannett}
\address{Mike Stannett, School of Computer Science, The University of Sheffield, Sheffield, UK}
\email{m.stannett@sheffield.ac.uk}

\author{Gergely Sz\'ekely}
\address{Gergely Sz\'ekely, HUN-REN Alfr\'ed R\'enyi Institute of Mathematics, Budapest, Hungary  \& University of Public Service, Budapest, Hungary.}
\email{szekely.gergely@renyi.hu}


\date{\today}

\thanks{This research is supported by the Hungarian National
  Research, Development and Innovation Office (NKFIH), grants
  no.\ FK-134732 and TKP2021-NVA-16.}

\begin{document}

\begin{abstract}
  In this paper, we prove a pure mathematical result which has
  important implications for the history and philosophy of classical
  physics and the conceptual origins of relativity theory. In formal
  terms, we show that, up to definitional equivalence, there is no
  intermediate model of spacetime lying strictly between special
  relativity and late classical kinematics. Informally, this means
  that there was essentially no other option but to switch to special
  relativity to resolve the conflict between late classical kinematics
  and the null result of the Michelson--Morley experiment.
\end{abstract}

\maketitle

\section{introduction}
In this paper, we prove a pure mathematical result which has important
implications for the history and philosophy of classical physics and
the conceptual origins of relativity theory.

The late 19th century was a period of transition. It had long been
believed that time was an absolute concept; and it had been known since
the 1700s that the speed of light was finite; but it was not until the
late 1860s, with the publication and refinement of Maxwell's
equations, that it was suggested that light could be interpreted in
terms of electromagnetic waves moving in all directions at lightspeed
through some ``luminiferous {\ae}ther''.  This gives rise to the idea
of \emph{lightlike-relatedness}, where we say that two spacetime
points are lightlike-related if it is feasible for a light signal to
travel from one to the other. This led for a while to a model of the
physical universe, which we refer to as \emph{late classical
spacetime} (or \emph{kinematics}).

Because light travels at the same speed in all directions within the
{\ae}ther, it made sense to think of the {\ae}ther as constituting an
absolute rest frame, and this fortuitously gave concrete meaning to Newton's, until then
``unphysical,'' concept of absolute space as something relative to 
which one can measure motion. This viewpoint was subsequently rendered untenable by the
Michelson--Morley experiment. So there appeared the need to find a
model of spacetime that gets rid of some classical concepts while keeping 
the notion of lightlike-relatedness.

Staying in harmony with the Galilean principle of relativity,
Einstein's solution, in some sense, was to keep \emph{only}
lightlike-relatedness and concepts definable from it.  But was this
minimalistic solution the only option left `on the table' or are there
other possible solutions -- spacetimes which, while not as loose as
relativity, are nonetheless not as constrained as late classical
spacetime itself? In another words, is there an intermediate theory
that lies conceptually between special relativity and late classical
kinematics?

This is the subject of Andr{\'e}ka's Conjecture. While discussing work
by Lefever and Sz{\'e}kely \cite{diss,ClassRelKin} -- which showed, in
an axiomatic framework, that one only needs to add a single concept,
absolute simultaneity (\ie absolute time), to special relativity to
get back to late classical kinematics up to definitional equivalence
-- Andr{\'e}ka conjectured in 2017 that this holds for \emph{any}
non-relativistic classical concept; in other words, there are no
intermediate spacetimes between relativistic and late classical, or at
least, not if we regard spacetime as a structure built on top of
4-dimensional real space, $\mathbb{R}^4$.

In other words, as Andr{\'e}ka conjectured and as we prove formally in
this paper, the need to retain lightlike-relatedness while accepting
the null outcome of the Michelson--Morley experiment, without
introducing entirely new concepts, meant that Einstein's relativistic
spacetime was the \emph{only} option left on the table.

Andr{\'e}ka originally formulated her conjecture in terms of Tarskian
algebraic logic, see Theorem~\ref{maximal}, but one does not need to be
familiar with algebraic logic to understand its content.  Here we
formulate and prove an equivalent version of the conjecture using
classical model theory and definability theory, see
Theorem~\ref{Hajnalconjecture}.

\subsection*{Notation and conventions}
In general, we use standard model theoretic and set theoretic
notations.  By a \emph{model}, we mean a first-order structure in the
sense of classical model theory. Given models $\Model$ and $\Nodel$,
we write $M$ and $N$, respectively, to denote their universes.

If $\rel\subseteq M^n$ is a relation over a nonempty set $M$ and
$f\colon M \to M$ is a map, we say that $\rel$ is \emph{closed under}
$f$ if{}f $(a_1,\dots, a_n)\in \rel \implies (f(a_1),\dots,
f(a_n))\in\rel$, and that $f$ \emph{respects} $\rel$ if{}f
$(a_1,\dots, a_n)\in \rel \iff (f(a_1),\dots, f(a_n))\in\rel$.  We
often write ``$\rel(a_1,\dots,a_n)$'' in place of ``$(a_1,\dots,a_n)
\in \rel$''.

We fix an enumeration $\fv{v}_1,\fv{v}_2,\fv{v}_3,\ldots$ of distinct
variables. If $(a_1,\ldots,a _n)\in M^n$ and $\varphi$ is a
first-order formula in the language of $\Model$, we write
$\Model\models\varphi[a_1,\ldots,a_n]$ to mean that $\varphi$ is
satisfied in $\Model$ by any evaluation
$e\colon\{\fv{v}_1,\fv{v}_2,\ldots\}\to M$ of variables for which
$e(\fv{v}_1)=a_1$, \ldots, $e(\fv{v}_n)=a_n$. The expression
$\varphi(\fv{v}_1,\ldots,\fv{v}_n)$ indicates that the free variables
of formula $\varphi$ come from the set $\{ \fv{v}_1,\ldots,\fv{v}_n
\}$.

We write $\aut{\Model}$ to denote the set of automorphisms of model
$\Model$, and note that, if the language of $\Model$ contains only
relation symbols, then a function $f\colon M\to M$ is an automorphism
of $\Model$ exactly if it is a bijection that respects all the
relations of $\Model$.

If $\Model$ is a model and $\rel$ is a relation on the
universe of $\Model$, then model $\la \Model,\rel\ra$ is the
expansion of $\Model$ with relation $\rel$ to some language that
contains exactly one extra relation symbol, whose interpretation in
$\Model$ is $\rel$.

We write $\realnumb$ for the set of real numbers, note that it forms a 
field when equipped with the usual operators and constants, and interpret 
$\realnumb^4$ as a vector space over $\realnumb$. We freely adopt standard 
notation for its associated functions (vector addition, multiplication of 
a vector by a scalar, etc.). We assume the reader is familiar with notions 
like \emph{linear transformation}.

We use both $A\subset B$ and $B\supset A$ to indicate that $A$ is a
proper subset of $B$.  The symbol $\Box$ indicates the end (or
absence) of a proof.

\subsection*{Concepts and definitional equivalence}

By a \emph{concept} of model $\Model$, we mean any relation definable
in $\Model$, where an $n$-ary relation $\rel$ on $M$ is \emph{definable} 
if{}f there is a first-order formula
$\varphi(\fv{v}_1,\ldots,\fv{v}_n)$ in the language of $\Model$ that
defines it; \ie for every $(a_1,\ldots,a_n)\in M^n$, we have
\begin{equation}\label{eq:def-def}
(a_1,\ldots,a_n) \in \rel \quad\Longleftrightarrow\quad
\Model\models\varphi[a_1,\ldots,a_n].
\end{equation}
We write $\ca{\Model}$ for the set of all concepts of $\Model$.

We say that models $\Model$ and $\Nodel$ are \emph{definitionally
equivalent} and write $\Model\defequiv\Nodel $ to mean they have the
same concepts, \ie $\ca{\Model}=\ca{\Nodel}$.\footnote{This
formulation of definitional equivalence is equivalent to the standard
formulation (\eg \cite[p.51]{HMT71} or \cite[p.453]{Monk2000}) but is
better suited to our requirements.}  Because the universe $M$ of
$\Model$ is definable in $\Model$ as a unary relation, we have
  \begin{equation*}
    \ca{\Model}\subseteq \ca{\Nodel} \implies M\subseteq N,
  \end{equation*}
  and hence
  \begin{equation*}
    {\Model}\defequiv {\Nodel} \implies M=N.
  \end{equation*}

\section{Relativistic and classical spacetimes}

We will define two models: \emph{relativistic spacetime} and
\emph{late classical spacetime}.  For simplicity, the points (or
events) of the spacetimes will be identified with $\realnumb^4$, and
if $(t,x,y,z)\in\realnumb^4$, we will call $t$ the time component and
$(x,y,z)$ the spatial component of $(t,x,y,z)$.  The binary relations
$\llcon$ of \emph{lightlike relatedness} and $\abstime$ of
\emph{absolute simultaneity} are defined on $\realnumb^4$ as follows:
\begin{eqnarray*}
\label{e-llcon}
(t,x,y,z)\,\llcon\, (t',x',y',z') & \defiff & (t-t')^2=(x-x')^2+(y-y')^2 +
(z-z')^2,\\
\label{e-abstime}
( t,x,y,z)\,\abstime\, (t',x',y',z') & \defiff & t=t'.
\end{eqnarray*} 

%%%%%%%%%%%%%%%%%%%%%%%%%%%%%%%%%%%%%%%%%%
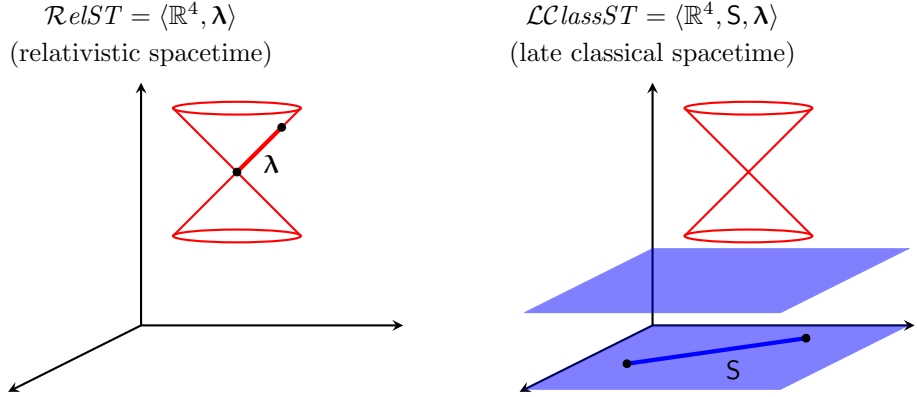
\begin{figure}[!htb]
  \begin{center}
    \begin{tikzpicture}[scale=1.7]

      \node[below] (R0) at (-2,4.2) {(relativistic spacetime)};
      \node[below] (N0) at (2,4.2) {(late classical spacetime)};

      \node[below] (R) at (-2,4.5) {$\RelST=\la \realnumb^4,\llcon\ra$};
      \node[below] (N) at (2,4.5) {$\LclST=\la \realnumb^4,\abstime,\llcon\ra$};

      \begin{scope}[shift={(-2,1.9)}]
        \draw[thick,<->, >=stealth] (0,1.9) to  (0,0) to (2.05,0) ;
        \draw [thick,->,>=stealth] (0,0) to (-1.04,-0.52);
        \draw[red,thick] (0.75,1.7) ellipse [x radius=0.5,y radius=0.05];
        \draw[red,thick] (0.25,1.7) to (1.25,0.7) (0.25,0.7) to (1.25,1.7);
        \draw[red,thick] (0.75,0.7) ellipse [x radius=0.5,y radius=0.05];
        \draw[red, ultra thick] (0.75,1.2) to node[black,below right=-2.5]  {$\llcon$} (1.1,1.55);
        \draw[fill] (0.75,1.2) circle [radius=0.03];
        \draw[fill] (1.1,1.55) circle [radius=0.03];
      \end{scope}

      \begin{scope}[shift={(2,1.9)}]
        \draw[thick, <->, >=stealth] (0,1.9) to  (0,0) to (2.05,0) ;
        \draw [thick, ->,>=stealth] (0,0) to (-1.04,-0.52);
        \draw[red,thick] (0.75,1.7) ellipse [x radius=0.5,y radius=0.05];
        \draw[red,thick] (0.25,1.7) to (1.25,0.7) (0.25,0.7) to (1.25,1.7);
        \draw[red,thick] (0.75,0.7) ellipse [x radius=0.5,y radius=0.05];
        \draw[fill,blue, opacity=0.5] (-1,-0.5) to (0,0) to (2,0) to (1,-0.5) to cycle;
        \draw[fill,blue, opacity=0.5] (-1,0.1) to (0,.6) to (2,0.6) to (1,0.1) to cycle;
        \draw[blue, ultra thick] (-0.2,-0.3) to node[black,below right=-.5]  {$\abstime$} (1.2,-.1);
        \draw[fill] (-0.2,-.3) circle [radius=0.03];
        \draw[fill] (1.2,-.1) circle [radius=0.03];
      \end{scope}

    \end{tikzpicture}
  \end{center}
  \caption{This figure illustrates relativistic spacetime $\RelST$ and late classical spacetime $\LclST$. \label{fig-st}}
\end{figure}
%%%%%%%%%%%%%%%%%%%%%%%%%%%%%%%%%%%%%%%%%%

We define relativistic and late classical spacetime, respectively, to be:
\begin{eqnarray*}
\RelST & \defeq & \la \realnumb^4,\llcon\ra,\\
\LclST & \defeq & \la \realnumb^4,\abstime,\llcon\ra,
\end{eqnarray*}
see Figure~\ref{fig-st}. These definitions may appear oversimplified,
but in Section~\ref{sec-lcs}, we will show that $\LclST$ is
definitionally equivalent to Galilean spacetime extended with
lightlike relatedness $\llcon$, and a comparison of $\RelST$-style
and coordinate-based axiomatizations is investigated in
\cite{AN14}.

Since $\LclST$ is obtained by adding $\abstime$ to $\RelST$, it is obvious that $\ca{\RelST}\subseteq\ca{\LclST}$. Importantly for our purposes, however, the inclusion is strict:

\begin{prop}
\label{prop-tartalmazas}
$\abstime\not\in\ca{\RelST}$, and hence,
$\ca{\RelST}\subset\ca{\LclST}$.
\end{prop}  
\begin{proof}
To prove that $\abstime\not\in\ca{\RelST}$, it is enough to find an
automorphism of $\RelST$ that does not respect $\abstime$. Using the
convention that sets the speed $c$ of light to be 1, any Lorentz boost
with nonzero speed $v$ is such an automorphism.
\end{proof}

\section{Andr{\'e}ka's Conjecture}

Andr{\'e}ka's conjecture can now be stated; it concerns concepts that lie in the gap between $\ca{\RelST}$ and $\ca{\LclST}$.

%%Recall that, by  definition $\Model$ and $\Nodel$ are definitionally
%%equivalent, in symbols $\Model\defequiv\Nodel$ if{}f
%%$\ca{\Model} = \ca{\Nodel}$.

\begin{theorem}[Andr\'eka's Conjecture]
\label{Hajnalconjecture}
For any concept $\concept\in\ca{\LclST}$ with 
$\concept\not\in\ca{\RelST}$,
\[
\la \RelST,\concept\ra  \defequiv  \LclST,
\]
\end{theorem}
\noindent
or equivalently,
\begin{theorem}[Reformulation of Andr\'eka's conjecture]
\label{nemletezik}
  There is no model $\Model$  for which
\[
\ca{\RelST}\subset\ca{\Model}\subset\ca{\LclST}.
\]
\end{theorem}
The cylindric-relativized set algebra obtained from model $\Model$ is
denoted by $\CA{\Model}$; see Monk~\cite{Monk2000}.  Roughly speaking,
$\CA{\Model}$ is an algebraic structure whose universe is essentially
the set of concepts of $\Model$, whose operations correspond to the
logical connectives, and among whose constants are ones that
correspond to logical \textit{true} and \textit{false}.  We note that
$\CA{\Model}$ is a subalgebra of $\CA{\Nodel}$ if{}f
$\big(\ca{\Model}\subseteq\ca{\Nodel} \text{ and } M=N\big)$, whence
the above formulations are equivalent to Andr{\'e}ka's original
formulation, viz.
\begin{theorem}[Original formulation of Andr\'eka's Conjecture]
\label{maximal}
$\CA{\RelST}$ is a maximal proper subalgebra of $\CA{\LclST}$.
 \end{theorem}

%%%%%%%%%%%%%%%%%%%%%%%%%%%%%%%%%%%%%%%%%%
\begin{figure}[!hbt]
  \begin{center}
    \begin{tikzpicture}[]
      \draw[thick] (0,0) ellipse [x radius=5,y radius=2];
      \draw[thick] (-0.9,-.2) ellipse [x radius=2.5,y radius=1.2];
      \draw[thick, loosely dotted] (-0.8,0) ellipse [x radius=3.5,y radius=1.6];

      \node at (1.5,1.2) {$\nexists$};
      \node[rotate=13] (R) at (-2,1.1) {$\ca{\RelST}$};
      \node[rotate=-11] (N) at (2.5,2) {$\ca{\LclST}$};

      \draw[fill] (2,0.2) circle [radius=.04] node[right] {$\concept$};

      \begin{scope}[shift={(-1,0)}]
        \draw[fill] (-.3,0) circle [radius=.04] node[left=2] {$\llcon$};
        \draw[red] (0,0.5) ellipse [x radius=0.5,y radius=0.05];
        \draw[red] (-0.5,0.5) to (0.5,-0.5) (-0.5,-0.5) to (0.5,0.5);
        \draw[red] (0,-0.5) ellipse [x radius=0.5,y radius=0.05];
      \end{scope}

      \begin{scope}[shift={(0.3,0.3)}]
        \draw[fill] (4.1,-0.2) circle [radius=.04] node[right] {$\abstime$};
        \draw[fill,blue, opacity=0.5] (3,0) to (4,0) to (3.5,-0.25) to (2.5,-0.25) to cycle;
        \draw[fill,blue, opacity=0.5] (3,-0.35) to (4,-0.35) to (3.5,-0.6) to (2.5,-0.6) to cycle;
      \end{scope}
    \end{tikzpicture}
  \end{center}
  \caption{This figure illustrates the relation between the concepts
    of relativistic spacetime $\RelST$ and late classical spacetime
    $\LclST$. \label{fig-venn}}
\end{figure}
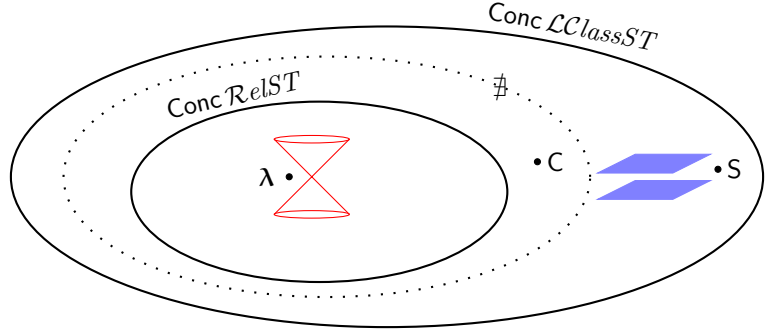
%%%%%%%%%%%%%%%%%%%%%%%%%%%%%%%%%%%%%%%%%%

\section{Sketch of the proof}
In this section, we give an essentially complete sketch of the proof
of Theorem~\ref{Hajnalconjecture}. Completing the details requires
proving a few supporting results, but we postpone these to later
sections.

We will define finitely field-definable (FFD) coordinate geometries
(Definition~\ref{FFD-def}), and show (Theorem~\ref{thm-erlangen}) that
for any two such geometries $\geometry$ and $\geometry'$,
\begin{align}
\label{egy}
\ca{\geometry}\subseteq\ca{\geometry'}\ \Longleftrightarrow &\
 \ \aut{\geometry}\supseteq\aut{\geometry'}, \\
\label{ketto}
\ca{\geometry}\subset\ca{\geometry'}\ \Longleftrightarrow &\ \
\aut{\geometry}\supset\aut{\geometry'}\ \text{ and }\\ 
\label{harom}
\ca{\geometry}=\ca{\geometry'}\ \Longleftrightarrow &\ \
\aut{\geometry}=\aut{\geometry'}.
\end{align}
We will see (Proposition~\ref{spacetimes-are-geometries}) that $\RelST$,
$\LclST$ and $\la \RelST,\concept\ra$ are all FFD
geometries. 

Clearly, from the assumptions of the conjecture (that $\concept\in\ca{\LclST}$ and
$\concept\not\in\ca{\RelST}$), we have that
\[
\ca{\RelST}\subset\ca{\la \RelST,\concept\ra}\subseteq\ca{\LclST},
\]
and we have already seen (Proposition~\ref{prop-tartalmazas}) that
these assumptions are not contradictory.  By \eqref{egy} and
\eqref{ketto}, this implies that
\begin{equation}
\label{milegyenaneve}
\aut{\RelST}\supset\aut{\la \RelST,\concept\ra}\supseteq\aut
{\LclST}.
\end{equation}
We will see (Corollary~\ref{nincskozte})
that there is no set of functions $\group$ that forms a group
under composition and satisfies
$\aut{\RelST}\supset\group\supset\aut{\LclST}$. 
This together with \eqref{milegyenaneve}
implies that $\aut{\la\RelST,\concept\ra}=\aut{\LclST}$; hence by
\eqref{harom}, $\la\RelST,\concept\ra\defequiv\LclST$ as required.

To turn the above sketch into a complete proof, we only need to
  show the supporting statements. We do this now; we begin by introducing coordinate geometries.

\section{Coordinate geometries}

The ternary relation $\Col$ of \emph{collinearity} on
the set of spacetime points $\realnumb^4$ is defined 
for every $\vec{p},\vec{q},\vec{r}\in\realnumb^4$ by
\[
\Col(\vec{p},\vec{q},\vec{r}\,)\quad \defiff\quad  
\vec{q}=\vec{p}+ a(\vec{r}-\vec{p}\,)\ \text{ for some }\ 
a\in\realnumb\text{, or }\vec{r}=\vec{p}.
\]
\begin{definition}
\label{geom-def}
A model $\geometry$ is called a ($4$-dimensional) 
\emph{coordinate geometry} if{}f the following conditions
are satisfied:
\begin{itemize}
\item the universe of $\geometry$ is $\realnumb^4$ 
(the set of \emph{points});
\item $\geometry$ contains no functions or constants;
\item
the ternary relation $\Col$ of collinearity 
on points is definable in $\geometry$.
\end{itemize}
\end{definition}

There is a natural correspondence between $n$-tuples of points in
$\realnumb^4$ ($n$-tuples of ($4$-tuples over $\realnumb$)) 
and $(4n)$-tuples over
$\realnumb$. Thus there is a natural correspondence
between $n$-ary relations on $\realnumb^4$ and
$4n$-ary relations on $\realnumb$. For
example, if $\rel$ is a ternary relation on
$\realnumb^4$, then the corresponding $(3\times 4)$-ary
relation $\widehat{\rel}$ 
on $\realnumb$ is defined by
\begin{multline*}
\widehat{\rel}(p_1,p_2,p_3,p_4,q_1,q_2,q_3,q_4,r_1,r_2,r_3,
r_4)
\defiff \\ \rel((p_1,p_2,p_3,p_4),
( q_1,q_2,q_3,q_4),(r_1,r_2,r_3, r_4))
\end{multline*}
for every $p_i,q_i,r_i\in\realnumb$ ($i\in\{1,2,3,4\}$). This example
can be generalized to $n$-ary relations straightforwardly.  If $\rel$
is an $n$-ary relation on $\realnumb^4$, then the corresponding
$4n$-ary relation on $\realnumb$ will be denoted by $\widehat{\rel}$.
For a precise definition of $\widehat{\rel}$, see \cite{part1}.

We say that an $n$-ary relation $\rel\subseteq
\left(\realnumb^4\right)^n$ is \emph{field-definable} if{}f the corresponding
$4n$-ary relation $\widehat{\rel}\subseteq {\realnumb}^{4n}$ is
definable in the field of reals $\la \realnumb, +,\cdot,0,1\ra$.

\begin{definition}
\label{FFD-def} 
We call coordinate geometry $\geometry$ 
\emph{finitely field-definable (FFD)} if{}f 
$\geometry$ contains only finitely many relations and
they are all field-definable.
\end{definition}

\begin{prop}\label{prop-fd}
The ternary relation $\Col$ 
and binary relations $\llcon$ and $\abstime$ on 
$\realnumb^4$ are field-definable.
\end{prop}
\begin{proof}
  It is straightforward to turn the definitions of $\Col$, $\llcon$
  and $\abstime$ into formulas on the language of fields that define
  $\widehat{\Col}$, $\widehat{\llcon}$ and $\widehat{\abstime}$ in
  $\la \realnumb, +,\cdot,0,1\ra$. For example, after reorganizing the
  equation to contain only basic operations and associating sequences
  of variables to points as $\vec{p}=(v_1,v_2,v_3,v_4)$,
  $\vec{q}=(v_5,v_6,v_7,v_8)$ and
  $\vec{r}=(v_9,v_{10},v_{11},v_{12})$, one gets formula
\begin{equation*}
\varphi_{\Col}(v_1,\ldots,v_{12})\defeq \exists
v_{13}\left(\bigwedge_{i=1}^4
\left(v_{4+i}+v_{13}\cdot v_i=v_i+v_{13}\cdot v_{8+i}\right) \right)\vee
\bigwedge_{i=1}^4 v_{8+i}=v_i
\end{equation*}
that defines $\widehat{\Col}$. Finding corresponding formulas for
$\widehat{\llcon}$ and $\widehat{\abstime}$ is similarly
straightforward.
\end{proof}

The following is a special case of  \cite[Thm.5.1.2]{part1}.
\begin{theorem}
\label{THM1}
Let $\geometry$ be an FFD coordinate geometry, and let $\rel$ be a
relation on $\realnumb^4$. Then $\rel\in\ca{\geometry}$ exactly
if $\rel$ is field-definable and closed under the automorphisms of
$\geometry$. \qed
\end{theorem}

\begin{prop}
\label{spacetimes-are-geometries}
$\RelST$, $\LclST$ and $\la \RelST,\concept\ra$ for all
$\concept\in\ca{\LclST}$ are FFD coordinate geometries.
\end{prop}
\begin{proof}
By Proposition~\ref{prop-fd}, in the case of $\RelST$ and $\LclST$,
the only thing that has to be proven is that $\Col$ is definable in
$\RelST$ and $\LclST$, which follows because $\Col$ can be defined
from $\llcon$, see, \eg \cite[\S 2.3]{Pambuccian07}. For the case $\la
\RelST,\concept\ra$, we need to show that $\concept$ is
field-definable, which follows from Theorem~\ref{THM1} because
$\concept$ is a concept of an FFD coordinate geometry.
\end{proof}

The following is a special case of \cite[Thm.5.1.4 and
  Cor.5.1.5]{part1}.

\begin{theorem}
\label{thm-erlangen}
Assume that  $\geometry$ and $\geometry'$ are
FFD coordinate geometries. 
Then
\begin{enumerate}[(i)]
\item $\ca{\geometry}\subseteq\ca{\geometry'}\ \Longleftrightarrow\
\aut{\geometry}\supseteq\aut{\geometry'}$, and hence,
\item  $\ca{\geometry}=\ca{\geometry'}\ \Longleftrightarrow\
\aut{\geometry}=\aut{\geometry'}$.
\item\label{item-proper}  $\ca{\geometry}\subset\ca{\geometry'}\ \Longleftrightarrow\
\aut{\geometry}\supset\aut{\geometry'}$. \qed
\end{enumerate}
\end{theorem}

\section{Transformations}

Here we introduce some sets of transformations and show their
connections with the automorphisms of the introduced spacetimes. We
use the ``$f$ after $g$'' order-convention for composition $f\circ g$
of functions, \ie $(f\circ g)(x)=f(g(x))$; and between sets of functions,
$\circ$ is understood as:
\begin{equation*}
  H\circ G \defeq \{\, h\circ g : h\in H \text{ and }g\in G\,\}.
\end{equation*}

The \emph{squared Euclidean length} and the \emph{squared Minkowski
length} of spacetime point $(t,x,y,z)\in\realnumb^4$ are defined as
\begin{align*}
\sqEl{(t,x,y,z)}
& \defeq  t^2+x^2+y^2+z^2,\\
\sqMl{(t,x,y,z)}
& \defeq  t^2-x^2-y^2-z^2.
\end{align*}

A map $L\colon \realnumb^4\to\realnumb^4$ is called a \emph{Lorentz
transformation} if{}f it is a bijective linear transformation that
preserves the squared Minkowski length, \ie $\sqMl{L(\vec{p}\,)}=
\sqMl{\vec{p}}$ for every $\vec{p}\in\realnumb^4$. A map $P\colon
\realnumb^4\to\realnumb^4$ is called a \emph{Poincar\'e
transformation} if{}f $P=\tau\circ L$ for some translation $\tau$ and
Lorentz transformation $L$. A map $A\colon\realnumb^4\to\realnumb^4$
is called a \emph{Euclidean isometry} if{}f $A=\tau\circ L$ for some
translation $\tau$ and bijective linear transformation $L$ that
preserves squared Euclidean length, \ie $\sqEl{L(\vec{p}\,)}=
\sqEl{\vec{p}}$ for every $\vec{p}\in\realnumb^4$.

Straight line $\ell$ is called \emph{vertical}, if it is parallel to
the time axis, \ie
\[
\ell=\{( t,x,y,z): t\in\realnumb\}\ \text{ for some
$x,y,z\in\realnumb$.}
\]

We call a map $A\colon \realnumb^4\to\realnumb^4$ a
\emph{trivial transformation} if it is a Euclidean isometry that
takes vertical lines to vertical lines.

The \emph{time-component} of point $(p_1,p_2,p_3,p_4)\in \realnumb^4$ is
defined by
\begin{equation}
  (p_1,p_2,p_3,p_4)_t\defeq p_1.
\end{equation}
A map  $A\colon
\realnumb^4\to\realnumb^4$ is \emph{orthochronous} iff $A$ `does not
change' the direction of time, \ie $A(1,0,0,0)_t>A(0,0,0,0)_t$.  And
$A$ is  a \emph{scaling} iff there is a nonzero $a\in\realnumb$
such that $A(\vec{p}\,)=a\cdot\vec{p}$.

The sets of Poincar\'e transformations and orthochronous Poincar\'e
transformations are denoted by $\Poi$ and $\Poi^{\uparrow}$,
respectively.  The sets of trivial transformations and orthochronous
trivial transformations are denoted by $\Triv$ and $\Triv^{\uparrow}$.
The set of scalings is denoted by $\Scal$.

We note that $\Poi$, $\Poi^{\uparrow}$, $\Triv$, $\Triv^{\uparrow}$
and $\Scal$ all form groups under composition.

\begin{theorem}[Alexandrov--Zeeman]
\label{spacetime-groups-a}
$\aut{\RelST}=\Scal\circ\Poi$.
\end{theorem}
\begin{proof}
The Alexandrov--Zeeman theorem \cite{Alexandrov} tells us that $\aut{\RelST} \subseteq \Scal\circ\Poi$. The reverse inclusion is straightforward.
\end{proof}

We are now going to determine $\aut{\LclST}$ and for this
we will state and prove several lemmas and propositions.

\begin{prop}
\label{fundamental}
Assume $\geometry$ is a coordinate geometry.
Then every element  of $\aut{\geometry}$ 
is a composition of a bijective linear transformation and a 
translation.
\end{prop}
\begin{proof}
Every automorphism of $\geometry$ is a bijection that respects $\Col$
since $\Col$ is definable in it; and it is known that any such
bijection has to be a linear transformation composed with a
translation, see \eg \cite[Thm.2, p.40]{Ryan86} or \cite[Exercise
  I.51, p.41]{Audin03}.
\end{proof}

The following relations and geometries are taken from our paper \cite{part2}. There we considered geometries of arbitrary dimension $d \geq 2$, defined over arbitrary ordered fields. The definitions and results stated here are specialized to $\realnumb^4$.

The binary relation $\rest$, 4-ary relation $\cng{}$ of 
Euclidean congruence and ternary relation
$\Bw$ of betweenness on points of $\realnumb^4$ are defined as
follows:
\begin{align*}
\nonumber
(t,x,y,z)\restrel (t',x',y',z')\quad\defiff\quad & (x,y,z)=(x',y',z'),\\
\nonumber
(\vec{p},\vec{q}\,)\cngrel{} (\vec{r},\vec{s}\,)\quad \defiff\quad & 
\sqEl{\vec{p}-\vec{q}}=\sqEl{\vec{r}-\vec{s}},\\
\label{Bw-def}
\Bw(\vec{p},\vec{q},\vec{r}\,)\quad \defiff\quad & 
\vec{q}=\vec{p}+a(\vec{r}-\vec{p}\,)\ \text{ for some }\ 
a\in [0,1] \subseteq \realnumb.
\end{align*}

We also recall two 
geometries from \cite{part2}, where the definitions and results are again specialized to $\realnumb^4$.

\begin{align*}
\euclg  &\ \defeq \ \la \realnumb^4,\cng{} ,\Bw\ra,\\
\relst  &\ \defeq \ \la \realnumb^4,\llcon,\Bw \ra.
\end{align*}

The following proposition connects the geometries $\euclg$ and $\relst$ to the spacetimes $\RelST$ and $\LclST$ used in the present paper.

\begin{prop}${}$
\label{connection}
\begin{enumerate}[(i)]
\item
\label{c1}
$\RelST\defequiv\relst$.
\item
\label{c2}
$\LclST\defequiv\la\euclg,\rest\ra$.
\end{enumerate}
\end{prop}
\begin{proof}
\underline{\eqref{c1}} We only have to prove that $\Bw$ is definable
in $\relst$.  This is easy to see, because the constraint ``$a\in [0 , 1]\subseteq\realnumb$'' in the definition of $\Bw$ can be
replaced by the equivalent statement ``$\exists b\, \exists c
\left(a=b^2\land (1-a)=c^2\right)$''.  Therefore, $\Bw$ is
field-definable.
 
By Propositions~\ref{spacetimes-are-geometries} and \ref{fundamental},
every automorphism of $\RelST$ is a bijective linear transformation
composed with a translation, and it is easy to see that $\Bw$ is
closed under linear transformations and translations. Therefore, $\Bw$
is closed under automorphisms of $\RelST$.  Applying
Theorem~\ref{THM1}, we see that $\Bw$ is definable in $\RelST$, hence
$\RelST\defequiv\relst$.

\underline{\eqref{c2}} By \cite[Thm.3.2.2]{part2}, we have
that $\abstime\not\in\ca{\relst}$ and $\rest\not\in\ca{\euclg}$.
Therefore, by \cite[Thm.3.2.3]{part2}, we have that
$\la\relst,\abstime\ra\defequiv\la\euclg,\rest\ra$.  By item
\eqref{c1}, we have that $\la \relst,\abstime\ra\defequiv \la
\RelST,\abstime\ra$. But $\la\RelST,\abstime\ra$ and $\LclST$  are essentially the
same, the only difference being the order in which the relations are listed.
Therefore, $\LclST\defequiv\la\euclg,\rest\ra$.
\end{proof}

\begin{theorem}\label{spacetime-groups-b}
$\aut{\LclST}=\Scal\circ\Triv$.
\end{theorem}
\begin{proof}
By Proposition~\ref{connection}\eqref{c2}, we have
that $\LclST\defequiv\la\realnumb^4,\cng{},\rest,\Bw\ra$,
and so 
\begin{equation}\label{eq-valami}
\aut{\LclST}=
\aut{\la\realnumb^4,\cng{},\rest,\Bw\ra}
\end{equation}
because automorphisms of definitionally equivalent models
coincide.

We will prove that 
\[
\aut{\la\realnumb^4,\cng{},\rest,\Bw\ra}
=\Scal\circ\Triv.
\]
It is easy to check
that all maps in $\Triv$ and $\Scal$ respect $\cng{}$, $\rest$ and $\Bw$, whence 
$\Scal\subseteq\aut{\la\realnumb^4,\cng{},\rest,\Bw\ra}$
and $\Triv\subseteq\aut{\la\realnumb^4,\cng{},\rest,\Bw\ra}$. So
$\Scal\circ\Triv\subseteq
\aut{\la\realnumb^4,\cng{},\rest,\Bw\ra}$.

To prove the reverse inclusion, suppose
$A\in\aut{\la\realnumb^4,\cng{},\rest,\Bw\ra}$.  By \eqref{eq-valami},
Propositions~\ref{fundamental} and \ref{spacetimes-are-geometries},
there is a translation $\tau$ and bijective linear transformation $L$
such that $A=\tau\circ L$.  Let
\[
\egyseg\defeq (1,0,0,0)  \quad\text{ and }\quad  \orrigo\defeq (0,0,0,0)
\]
and define $\|\vec{p}\,\| \defeq\sqrt{\sqEl{\vec{p}}}$ for $\vec{p} \in \realnumb^4$.

Observe that $L(\egyseg\,) \neq \orrigo$ because $L$ is bijective and linear. 
So we can safely define a map $T\colon\realnumb^4\rightarrow\realnumb^4$
by $T(\vec{p}\,)=\frac{1}{\|L(\egyseg\,)\|}L(\vec{p}\,)$,
and let $D$ be a scaling with factor $\|L(\egyseg\,)\|$,
\ie $D(\vec{p}\,)=\|L(\egyseg\,)\|\vec{p}$. Then
$T$ is a bijective linear transformation and $L=D\circ T$. 
So $A=\tau\circ D\circ T$.
As $A,\tau,D\in\aut{\la\realnumb^4,\cng{},\rest,\Bw\ra}$, 
we have that 
\[
T\in\aut{\la\realnumb^4,\cng{},\rest,\Bw\ra}.
\]
We will prove that $T\in\Triv$. First we will prove that
$T$ is a Euclidean isometry, by proving that
for every $\vec{p}\in\realnumb^4$,
$\|T(\vec{p}\,)\|=\|\vec{p}\,\|$. 
If $\vec{p} = \orrigo$, this is clearly true because $T(\orrigo\,) = \orrigo$, 
so suppose  that $\vec{p} \neq \orrigo$.
Clearly, 
$(\egyseg,\orrigo\,)\cngrel{}\left(\frac{1}{\|\vec{p}\,\|}\vec{p},
\orrigo\right)$.
As $T$ respects $\cng{}$, we have that
$(T(\egyseg\,),T(\orrigo\,))\cngrel{}\left(T\left(\frac{1}{\|\vec{p}\,\|}\vec{p}\right)
, T(\orrigo\,)\right)$. This implies by linearity that
$(T(\egyseg\,),\orrigo\,)\cngrel{}\left(\frac{1}{\|\vec{p}\,\|}
T(\vec{p}\,),\orrigo\,)\right)$. Then by definition of
$\cng{}$, we have that 
$\|T(\egyseg\,)\|= \|T(\egyseg\,)-\orrigo\,\|=\|\frac{1}{\|\vec{p}\,\|}
T(\vec{p}\,)-\orrigo\,\|=\frac{\|T(\vec{p}\,)\|}{\|\vec{p}\,\|}$.
On the other hand $\|T(\egyseg\,)\|
=\left\|\frac{1}{\|L(\egyseg\,)\|}L(\egyseg\,)\right \|
=1$. Therefore, $\frac{\|T(\vec{p}\,)\|}{\|\vec{p}\,\|}=1$, \ie
$\|T(\vec{p}\,)\|=\|\vec{p}\,\|$, and this implies that
$T$ is a Euclidean isometry. Since the bijective linear transformation
 $T$ respects $\rest$, it also takes vertical lines to vertical
lines. Thus $T\in\Triv$ as claimed.

Finally, let $\tau'$ be the translation
with vector $D^{-1}\left(\tau(\orrigo\,)\right)$. Then it is easy
to check that $\tau\circ D=D\circ\tau'$. Since $\tau'$ and $T$ are both trivial transformations, so is $\tau'\circ T$. So $A=\tau\circ D\circ T=
D\circ(\tau'\circ T)\in\Scal\circ\Triv$. 
\end{proof}

By Theorems~\ref{spacetime-groups-a} and \ref{spacetime-groups-b},
Proposition~\ref{prop-tartalmazas} and 
Theorem~\ref{thm-erlangen}\eqref{item-proper}, we have the following.
\begin{corollary}
\label{csop-cor}
$\Scal\circ\Triv\subset
\Scal\circ\Poi$ and both
$\Scal\circ\Poi$ and $\Scal\circ\Triv$ form 
groups under composition. \qed
\end{corollary}

\begin{prop} 
\label{prop-triv}
$\Triv\subset\Poi$
and 
$\Triv^{\uparrow}\subset\Poi^{\uparrow}$.
\end{prop}
\begin{proof}
By Corollary~\ref{csop-cor} and the fact that
the identity function on $\realnumb^4$ is a scaling,
we have that
\begin{equation}
\label{eee}
\Triv\subseteq\Scal\circ\Triv\subset\Scal\circ\Poi.
\end{equation}
Let
\[
\egyseg\defeq (1,0,0,0).
\]
To prove $\Triv\subseteq\Poi$, let $T\in\Triv$.  Then $T=\tau\circ L$
for some bijective linear transformation $L$ and translation
$\tau$. Clearly, $L\in\Triv$; and by \eqref{eee}, $L=D\circ P$ for
some $D\in\Scal$ and $P\in\Poi$. We also have that $P$ is a linear
transformation and that it maps vertical lines to vertical lines
because both $L$ and $D$ have these properties. As $P$ fixes the
origin, we have that $P(\egyseg\,)=(t,0,0,0)$ for some
$t\in\realnumb$.  We know that $P$ preserves the squared Minkowski
length, thus
$1=\sqMl{\egyseg}=\sqMl{P(\egyseg\,)}=\sqMl{(t,0,0,0)}=t^2$.
Therefore, $t=\pm 1$. This implies that either $P(\egyseg\,)=\egyseg$
holds or $P(\egyseg\,)=-\egyseg$ holds.  The same argument, this time
using the fact that $L$ preserves squared Euclidean length, shows that
either $L(\egyseg\,)=\egyseg$ or $L(\egyseg\,)=-\egyseg$.

Because $D\in\Scal$, $L(\egyseg\,)=\pm\egyseg$,
$P(\egyseg\,)=\pm\egyseg$ and $L=D\circ P$, we get that
the scaling factor of $D$ is either $1$ or $-1$, \ie
either $D(\vec{p}\,)=\vec{p}$ for every $\vec{p}$ or
$D(\vec{p}\,)=-\vec{p}$ for every $\vec{p}$. In both cases,
$D\in\Poi$. Therefore, $L=D\circ P\in\Poi$ as $\Poi$ forms a group.
As $\tau$ is a translation, we also have that
$\tau\in\Poi$. Thus $A=\tau\circ L\in\Poi$, and this
completes the proof that $\Triv\subseteq\Poi$.
This then implies easily that 
$\Triv^{\uparrow}\subseteq\Poi^{\uparrow}$. 

Taking the speed $c$ of light to be $1$, for any Lorentz boost $B$
with nonzero speed $v$, we have that $B\in\Poi^{\uparrow}$ but
$B\not\in\Triv^\uparrow$, and likewise $B\in\Poi$ and $B\not\in\Triv$.
Therefore, $\Triv^{\uparrow}\subset\Poi^{\uparrow}$ and
$\Triv\subset\Poi$ hold, as required.
\end{proof}

The following theorem is a special case of
\cite[Cor.4.9]{OnBorisov}, which is based on Borisov's theorem
\cite{Borisov1986}.

\begin{theorem}[Corollary of Borisov's theorem]${}$
\label{Borisov}
There is no set of functions $\group$ which forms a group under
composition satisfying
$\Triv^{\uparrow}\subset \group\subset
\Poi^{\uparrow}$. \qed
\end{theorem}
\begin{theorem}${}$
\label{Borisov2}
There is no set of functions $\group$ which forms
a group under composition satisfying
 $\Scal\circ\Triv\subset \group\subset\Scal\circ\Poi$.
\end{theorem}

\begin{proof}
First we show that
\begin{equation}\label{new1}
  (\Scal\circ\Triv)\cap\Poi^{\uparrow}=\Triv^{\uparrow}.
\end{equation}
This can be proven as follows. We have
$(\Scal\circ\Triv)\cap\Poi^{\uparrow}\supseteq\Triv^{\uparrow}$ since
$\Poi^{\uparrow}\supseteq \Triv^{\uparrow}$ by
Proposition~\ref{prop-triv} and $\Scal\circ\Triv\supseteq
\Triv^{\uparrow}$ as the identity map is a scaling. To prove the
reverse inclusion, it is enough to show
$(\Scal\circ\Triv)\cap\Poi\subseteq\Triv$. Let $f\in
(\Scal\circ\Triv)\cap\Poi$, and let $s\in \Scal$ and $t\in\Triv$ be
such that $f=s\circ t$. Since $\Scal\circ\Triv$, $\Triv$ and $\Poi$
are all groups containing translations, we can assume without loss of
generality that $f$ and $t$ are linear. Then $t$ maps $(1,0,0,0)$ to
either $(1,0,0,0)$ or $(-1,0,0,0)$, and the scaling factor of $s$ is
either $1$ or $-1$ because otherwise $s\circ t$ would not preserve
squared Minkowski length, and hence $f=s\circ t$ would not be a
Poincar\'e transformation. So $s\in\Triv$, and hence $f=s\circ
t\in\Triv$ because $\Triv$ is closed under composition.

We will also use the following statement:
\begin{equation}  \label{new2}
  \Scal\circ\Poi=\Scal\circ\Poi^{\uparrow}.
\end{equation}
To show \eqref{new2}, it is enough to show that
$\Scal\circ\Poi\subseteq \Scal\circ\Poi^{\uparrow}$ since
$\Poi\supseteq \Poi^{\uparrow}$ holds by the definition. So let
$f\in\Scal\circ\Poi$. Then $f=s\circ p$ for some $s\in\Scal$ and
$p\in\Poi$. Without loss of generality, we can assume that
$p(0,0,0,0)=(0,0,0,0)$ since translations are orthochronous Poincar\'e
transformations and $\Poi$ and $\Poi^\uparrow$ are groups under
composition. Let $s_{-1}$ denote the scaling by factor $-1$.  Then
either $p\in\Poi^\uparrow$ or $s_{-1}\circ p\in \Poi^\uparrow$ because
$p(1,0,0,0)_t$ is either positive or negative and $s_{-1}$ changes the
sign of the time-components of non-horizontal vectors. In case
$p\in\Poi^\uparrow$, we are done. In the other case, $f=(s\circ
s_{-1})\circ (s_{-1}\circ p)$ is a desired decomposition of $f$ since
$s\circ s_{-1}\in \Scal$ and $s_{-1}\circ p\in \Poi^\uparrow$.

Finally, to prove the main claim, let us assume that $\group$ forms a
group and that
\begin{equation}
\label{Be3}
\Scal\circ\Triv\subset\group\subseteq\Scal\circ\Poi.
\end{equation}
We will prove that $\group=\Scal\circ\Poi$. First let us note that by
\eqref{Be3} and the fact that $\Triv$ contains the identity function
on $\realnumb^4$, we have that $\Scal\subseteq\Scal\circ
\Triv\subset\group$.  Let $g\in \group$ such that
$g\not\in\Scal\circ\Triv$. By \eqref{Be3} and \eqref{new1}, this $g$
can be decomposed as $g=s\circ p^\uparrow$ for some $s\in\Scal$ and
$p^{\uparrow}\in\Poi^{\uparrow}$, \ie $s\circ p^{\uparrow}\in\group$
but $s\circ p^{\uparrow}\not\in\Scal\circ\Triv$.  Clearly,
$s\in\group$ by $s\in\Scal\subseteq\group$.  As $\group$ forms a group
and both $s$ and $s\circ p^{\uparrow}$ are in $\group$, we have that
$p^{\uparrow}\in\group$. On the other hand,
$p^{\uparrow}\not\in\Scal\circ\Triv$ because we would otherwise have
$s\circ p^{\uparrow}\in\Scal\circ(\Scal\circ\Triv) = \Scal\circ\Triv$
because $s\in\Scal$ and $\Scal$ is closed under composition; and we
have already seen that this is not the case.  Thus
$p^{\uparrow}\in\group\cap\Poi^{\uparrow}$ and $p^{\uparrow}\not\in
(\Scal\circ\Triv)\cap\Poi^{\uparrow}$.  By this and \eqref{Be3}, we
get that
\begin{equation}\label{eqx}
  (\Scal\circ\Triv)\cap\Poi^\uparrow\;
\subset\;\group\cap\Poi^{\uparrow}\;\subseteq\;
(\Scal\circ\Poi)\cap\Poi^{\uparrow}.
\end{equation}
As $\Poi^{\uparrow}\subseteq\Poi\subseteq\Scal\circ\Poi$, we have that
$(\Scal\circ\Poi)\cap\Poi^{\uparrow}=\Poi^{\uparrow}$.  By this,
\eqref{new1} and \eqref{eqx}, we have that
$\Triv^{\uparrow}\;\subset\;\group\cap
\Poi^{\uparrow}\subseteq\Poi^{\uparrow}$.  Now applying
Theorem~\ref{Borisov}, we get that
$\group\cap\Poi^{\uparrow}=\Poi^{\uparrow}$, and so
$\Poi^{\uparrow}\subseteq\group$. 

We also have that
$\Scal\subseteq\group$. Therefore,
$\Scal\circ\Poi^{\uparrow}\subseteq\group$.  As
$\Scal\circ\Poi^{\uparrow}=\Scal\circ\Poi$ by \eqref{new2}, we get
that $\Scal\circ\Poi\subseteq\group$. The reverse inclusion holds by assumption \eqref{Be3}.

So $\group=\Scal\circ\Poi$ as
stated.
\end{proof}

The following is a corollary of Theorems~\ref{spacetime-groups-a},
\ref{spacetime-groups-b} and \ref{Borisov2}:
\begin{corollary}
\label{nincskozte}
There is no set of functions $\group$ that forms a group under
composition, satisfying
$\aut{\LclST}\subset\group\subset\aut{\RelST}$.
\end{corollary}

\section{Why is $\LclST$ late classical spacetime?}
\label{sec-lcs}

Here we justify the name late classical spacetime by showing that
$\LclST$ is definitionally equivalent to Galilean spacetime extended
with lightlike relatedness $\llcon$. To do so, let us define the 4-ary
relations on $\realnumb^4$ of \emph{spatial congruence}
$\cng{\abstime}$ and \emph{temporal congruence} $\cng{\mathsf{T}}$ as:
\begin{align*}
(\vec{p},\vec{q}\,)\cngrel{\abstime} (\vec{r},\vec{s}\,) \quad \defiff\quad &
  \ \vec{p}\,\abstime\,\vec{q},\ \ \vec{r}\,\abstime\,\vec{s}
  \ \text{ and }\
  (\vec{p},\vec{q}\,)\cngrel{}(\vec{r},\vec{s}\,),\\ (\vec{p},\vec{q}\,)\cngrel{\mathsf{T}}
  (\vec{r},\vec{s}\,) \quad \defiff\quad &
  \vec{p}_t-\vec{q}_t=\vec{r}_t-\vec{s}_t,
\end{align*}
and introduce Galilean spacetime as:
\[
\GalST\defeq\la \realnumb^4,\cng{\abstime},\cng{\mathsf{T}},
\Col\ra.
\]
We note that temporal congruence $\cng{\mathsf{T}}$ is definable in
terms of $\abstime$ and $\Col$; and this implies that temporal
congruence $\cng{\mathsf{T}}$ is definable in $\la
\realnumb^4,\cng{\abstime},\Col\ra$ because $\abstime$ is definable in
terms of $\cng{\abstime}$. For an axiomatic approach to Galilean spacetime using similar languages, see \eg \cite{Field80} and \cite{ketland23}.

\begin{prop}
$\la\GalST,\llcon\ra\defequiv\LclST$.
\end{prop}
\begin{proof}
Recall that $\LclST=\la\realnumb^4,\abstime,\llcon\ra$, and note that $\abstime$ is definable in terms of
 $\cng{\abstime}$ as $\vec{p}\,\abstimerel\,\vec{q}\ \Leftrightarrow\
(\vec{p},\vec{q}\,)\cngrel{\abstime}(\vec{p},\vec{q}\,)$. Therefore,
\[
\ca{\la\GalST,\llcon\ra}\supseteq \ca{\LclST}.
\]

Next, note that $\Triv\subseteq\aut{\la\GalST,\llcon\ra}$ and
$\Scal\subseteq\aut{\la \GalST,\llcon\ra}$.  This follows easily
because $\la\GalST,\llcon\ra$ is $\la
\realnumb^4,\cng{\abstime},\cng{\mathsf{T}}, \Col,\llcon\ra$, and all
maps in $\Triv$ and $\Scal$ respect
$\cng{\abstime}$,$\cng{\mathsf{T}}$, $\Col$ and $\llcon$.  Therefore,
$(\Scal\circ\Triv)\subseteq\aut{\la \GalST,\llcon\ra}$.  But, by
Theorem~\ref{spacetime-groups-b}, $\aut{\LclST}=\Scal\circ\Triv$.
Consequently, $\aut{\LclST}\subseteq\aut{\la \GalST,\llcon\ra}$,
whence by Theorem~\ref{thm-erlangen} 
\[
\ca{\la\GalST,\llcon\ra} \subseteq \ca{\LclST}
\]
because $\la\GalST,\llcon\ra$ is
also an FFD coordinate geometry.
Thus $\ca{\LclST}=\ca{\la\GalST,\llcon\ra}$,
\ie $\LclST\defequiv\la\GalST,\llcon\ra$.
\end{proof}

\bibliographystyle{amsalpha}
\bibliography{LogRel12019}

\end{document}